\documentclass[prb,aps,reprint,twocolumn,nofootinbib,showpacs,groupedaddress,floatfix]{revtex4-2}

\usepackage{lmodern}
\usepackage{microtype}

\usepackage{amsmath,amssymb,amsthm,mathtools,bm}
\usepackage{mathrsfs}   % for M in appendix

\usepackage{bbm} % for \mathbbm{1} if you want an indicator / identity
\usepackage{graphicx}
\usepackage{booktabs}
\usepackage{array}

\usepackage[dvipsnames]{xcolor} %
\usepackage{hyperref}
\hypersetup{colorlinks=true,allcolors=YellowOrange,pdfborder={0 0 0}}
\theoremstyle{plain}
\newtheorem{theorem}{Theorem}
\newtheorem{lemma}{Lemma}

\theoremstyle{definition}

\theoremstyle{remark}

\newcommand{\cP}{\mathcal{P}}

\DeclareMathOperator{\supp}{supp}

\renewcommand\vec{\boldsymbol}

\def\r{{\vec{r}}}

\def\k{{\vec{k}}}

\def\R{{\vec{R}}}

\newcommand{\blue}[1]{ #1}

\begin{document}

\title{Bootstrapping ground state properties of classical frustrated magnets}
\author{Nisarga Paul}
\email{npaul@caltech.edu}
\affiliation{Department of Physics and Institute for Quantum Information and Matter,
California Institute of Technology, Pasadena, California 91125, USA}
\author{Gil Refael}
\email{refael@caltech.edu}
\affiliation{Department of Physics and Institute for Quantum Information and Matter,
California Institute of Technology, Pasadena, California 91125, USA}
\begin{abstract}
We introduce a method based on semidefinite programming that produces rigorous bounds on ground state energy densities and correlation functions of translation-invariant classical spin models on infinite lattices. In this method, the challenge of non-convex optimization on an infinite lattice is replaced with a hierarchy of finite-size convex optimizations arising from positivity conditions that any probability distribution over spin configurations must satisfy. This adapts the Lasserre hierarchy in the theory of polynomial optimization to the context of frustrated magnetism, and we prove convergence of this hierarchy in the thermodynamic limit. Our method subsumes the Luttinger--Tisza method and applies to non-quadratic Hamiltonians and non-Bravais lattices, thus addressing limitations of prior analytical methods. We apply the method to various two-dimensional frustrated spin models, where it brackets the energy densities and observables accurately across large parameter ranges with typical run times of seconds per parameter point.  
\end{abstract}

\maketitle
\section{Introduction}

Frustrated magnets, in which competing interactions cannot be simultaneously satisfied, host some of the richest phenomena in condensed matter physics~\cite{lacroix2011introduction,balents2010spin}. Much of this richness lives in the ground state, which can exhibit extensive degeneracy, order by disorder~\cite{villain1980obdo}, and emergent gauge fields~\cite{moessner1998properties,huse2003coulomb,castelnovo2012spin,rehn2016classical,henley2010coulomb,chalker1992hidden}. Paradigmatic examples include the triangular-lattice Ising antiferromagnet~\cite{wannier1950antiferromagnetism} and spin ice on the pyrochlore lattice~\cite{bramwell2001spin}. Moreover, classical frustrated magnetism is relevant to quantum spin systems both at large spin and as a starting point for understanding frustration effects and order parameters.
\par 
Determining the ground state spin configuration of a frustrated spin system is a formidable optimization problem, known to be \textsf{NP}-hard for Ising spins~\cite{barahona1982computational,lucas2014ising}. For Heisenberg spins with translation-invariant Hamiltonians, the most widely used analytical approach is the Luttinger-Tisza (LT) method~\cite{luttinger1946theory,lyons1960method}. The LT method, which minimizes the energy over sums-of-spirals ansatzes, has limited validity and various limitations. The main numerical alternative, classical Monte Carlo, can face severe equilibration challenges in frustrated systems where the energy landscape has many closely competing local minima~\cite{hukushima1996exchange,earl2005parallel}.  

Bootstrap methods, or semidefinite relaxation methods, have proven useful in a variety of physics applications across high energy theory~\cite{el2012solving,lin2020bootstraps,han2020bootstrap,lin2025high}, condensed matter~\cite{anderson1951limits,barthel2012sdp,gao2025bootstrapping,scheer2025defect}, and quantum information~\cite{doherty2004complete,navascues2007bounding,navascues2008convergent}. In the condensed matter context, they have their root in Anderson-type lower bounds on ground state energies of antiferromagnets, wherein each term of the Hamiltonian is minimized separately~\cite{anderson1951limits}. Since then, these methods have grown in sophistication and capitalized on advances in convex optimization. The basic idea in these approaches is that operator expectation values in a true ground state density matrix satisfy infinitely many algebraic constraints, and these constraints can be solved self-consistently to determine all physical quantities. Bootstrap methods have found application in the quantum Hall problem, matrix quantum mechanics, quantum chemistry, the Sachdev-Ye-Kitaev model, and Ising models, to name a few~\cite{coleman1963rdm,nakata2001rdm,mazziotti2004rdm,haim2020variational,gao2025bootstrapping,han2020bootstrap,hastings2022optimizing,scheer2025defect}. More recently, semidefinite methods have also been developed to obtain rigorous bounds on ground state energies and local observables of quantum lattice models directly in the thermodynamic limit~\cite{barthel2012sdp,kull2024lower,fawzi2024certified,cho2026coarse}. However, it appears that bootstrap methods have not been fully brought to bear on ground states of \textit{classical} frustrated magnets with continuous spins, although related approaches exist for Ising systems~\cite{cho2023bootstrap}.

\begin{figure}
    \centering    \includegraphics[width=\linewidth]{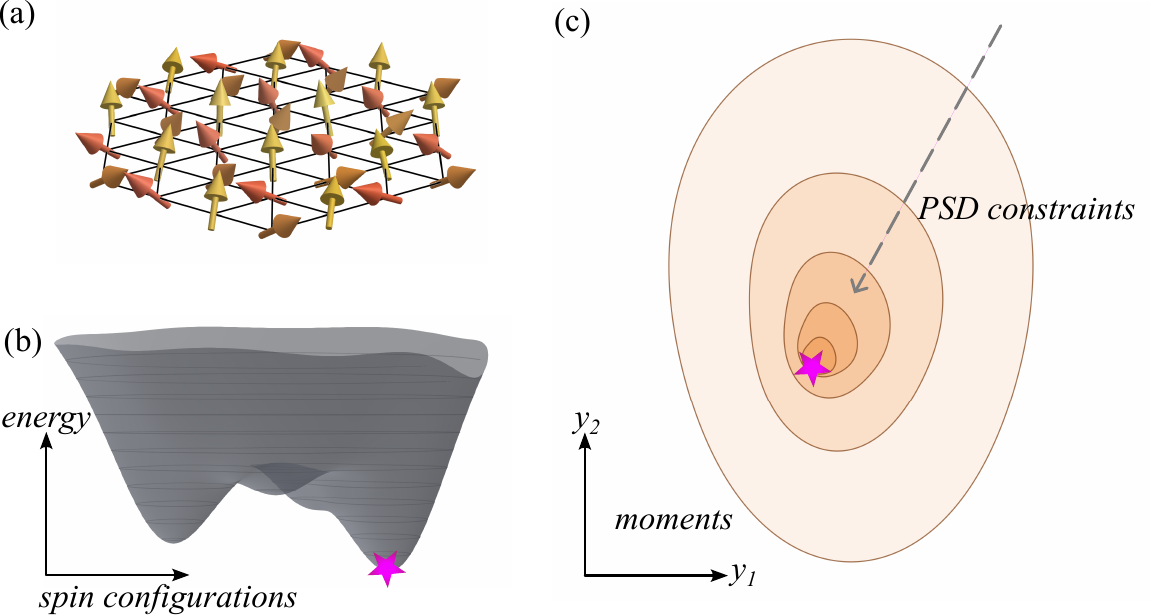}
    \caption{\textit{Bootstrap approach to classical frustrated magnetism.} (a) Schematic of spin texture arising from frustrated interactions. (b) Determining the ground state spin configuration and energy is a non-convex optimization problem with multiple competing local minima. (c) The bootstrap replaces this with a hierarchy of optimization problems over convex domains (shaded regions) in the space of moments of spin variables carved out by positive semidefinite (PSD) constraints. The regions shrink monotonically toward the true solution (star) as more PSD constraints are imposed.}
    \label{fig:1}
\end{figure}

In this work, we introduce a bootstrap method to solve ground state problems in classical frustrated magnetism (Fig.~\ref{fig:1}). One approach could be to take recent advances in quantum many-body bootstrap and replace the quantum Hamiltonian with a trivially commuting one. However, this is not the approach taken here. Instead, our starting point is the observation that finding the ground state of a classical frustrated magnet reduces to a polynomial optimization problem with polynomial constraints enforcing unit-norm spins, which for finitely many variables can be efficiently solved by a bootstrap method known as Lasserre's hierarchy~\cite{lasserre2001global} (or the moment / sum-of-squares hierarchy~\cite{lasserre2018moment}). We review Lasserre's hierarchy in Appendix \ref{app:lasserre} for the interested reader but keep our treatment in the main text self-contained. In this work, we generalize Lasserre's hierarchy to infinitely many variables and adapt it to the case of translation-invariant classical spin Hamiltonians.  

\par 

The result is a bootstrap method which produces rigorous bounds on the ground state energy density and arbitrary ground state spin correlation functions directly in the thermodynamic limit. This distinguishes it from variational or Monte Carlo methods, which may be finite-size and do not come with guarantees on proximity to the true ground state. Our approach also contains the Luttinger-Tisza method as a subcase, but extends it to non-quadratic Hamiltonians and non-Bravais lattices. Moreover, it can be efficiently implemented using a numerical semidefinite program (SDP) solver such as \texttt{MOSEK}~\cite{mosek}.\par 

The core idea is to replace the challenging non-convex optimization problem over spin configurations by a hierarchy of convex optimizations over translation-invariant correlation functions, subject to positivity and normalization constraints that any valid correlation functions must satisfy. This hierarchy is labeled by a real-space patch $\mathcal{P}$ and a positive integer $r$. As $\cP$ and $r$ grow, the energy lower bounds converge monotonically to the exact thermodynamic-limit ground state energy. The proof of this convergence combines tools from measure theory and real algebraic geometry, and we present it in Appendix~\ref{app:convergence}.

We demonstrate the method on two classical frustrated spin models in two dimensions. The first is a frustrated Heisenberg model on a lattice with non-uniform coordination, where the LT method fails. The second is a highly frustrated quartic model on the triangular lattice (also known as the bilinear-biquadratic model). 
Across broad parameter ranges, the SDP and variational energy bounds agree to within $10^{-5}$--$10^{-6}$ per site (in units of the coupling). Beyond bounds on energies and correlation functions, the bootstrap produces explicit spin textures extracted from the SDP solution, which match the variational ground states well. Moreover, solving each SDP takes only a few seconds on a single CPU. We thus propose that our bootstrap approach can serve as a practical general-purpose tool for ground state problems in classical frustrated magnetism complementary to existing heuristic approaches and Monte Carlo.  

\par

\section{Method}
\label{sec:method}

We develop the method in two stages. We first treat the case of quadratic spin models in Sec.~\ref{sec:quadratic}. This case contains the Luttinger--Tisza method as a special case (Sec.~\ref{sec:LT}) and demonstrates the main ideas of the SDP approach. We then treat the general case in Sec.~\ref{sec:spin_sdp}, and finally discuss how to obtain bounds on correlation functions, obtain upper bounds and extract explicit spin textures in Sec.~\ref{sec:observables}.

\subsection{Quadratic Heisenberg Hamiltonians}
\label{sec:quadratic}

For simplicity, we start by considering quadratic Heisenberg Hamiltonians. We consider a periodic lattice in $d$ dimensions with $n_s$ sites per unit cell. Each site is labeled by $(\R, \mu)$, where $\R$ is a Bravais lattice vector and $\mu \in \{1, \ldots, n_s\}$ is a sublattice index.
At each site we place a classical spin $\vec{S}_{\R,\mu} = (S_{\R,\mu}^x, S_{\R,\mu}^y, S_{\R,\mu}^z)$ with unit norm $|\vec{S}_{\R,\mu}|^2 = 1$. All of our results can be straightforwardly generalized to account for Ising spins, for instance by enforcing $S^x = S^y=0$, or indeed any classical degree of freedom living on a compact space. We assume a Hamiltonian of the form
\begin{equation}
    H = \sum_{\R\R'}\sum_{\mu\nu}
    J_{\mu\nu}(\R'-\R)\,\vec S_{\R,\mu}\cdot \vec S_{\R',\nu}, 
    \label{eq:hamiltonian}
\end{equation}
In the thermodynamic limit, we may consider instead the energy density, i.e. energy per unit cell. Taking the reference site $\R=\vec 0$, we define an energy density
\begin{equation}
    \mathcal{E} =  \frac{1}{n_s} \sum_{\mu, \nu,\R} J_{\mu\nu}(\R)\, \vec S_{\vec 0,\mu}\cdot \vec S_{\R,\nu}.
    \label{eq:bilinear_Phi}
\end{equation}
Next, we observe that the expectation value of $\mathcal{E}$ depends only on the two-point correlation function. In particular, 
\begin{equation}
    \varepsilon \coloneqq \langle \mathcal{E} \rangle = \frac{1}{n_s} \sum_{\mu,\nu,\R} J_{\mu\nu}(\R)\, C_{\mu\nu}(\R)
    \label{eq:energy_multisub}
\end{equation}
where $C_{\mu\nu}(\R) \coloneqq \langle \vec S_{\vec 0,\mu}\cdot \vec S_{\R,\nu}\rangle$. Here, the expectation value denotes integration against a distribution over spin configurations; in the case of a single spin configuration, it is simply the evaluation of the expression in that spin configuration. Instead of minimizing the energy over all spin configurations, we can minimize the linear functional Eq.~\eqref{eq:energy_multisub} over the set of valid two-point correlation functions.  
\par 
The challenge now becomes to characterize the set of valid two-point correlation functions, i.e. to characterize which $C_{\mu\nu}(\R)$ can arise from valid spin configurations. Below, we derive two necessary conditions. The first necessary condition is \textit{positivity}: consider a finite patch of lattice sites $\cP$ and an arbitrary real linear combination $f$ of the spin variables in this patch. We must have $\langle f^2 \rangle \geq 0$ for any spin configuration. Expanding the square and using translation and $O(3)$-invariance, we find that the matrix defined by
\begin{equation}
\mathcal{G}_{ij} = C_{\mu\nu}(\R' - \R), \quad i = (\R,\mu),\;\; j = (\R',\nu)
\end{equation}
is positive semidefinite, i.e. $\mathcal{G}\succeq 0$. The second necessary condition is \textit{normalization}: the unit-length constraint $|\vec{S}_{\R,\mu}|^2 = 1$ imposes 
 \begin{equation}
C_{\mu\mu}(\vec 0) = 1 
\label{eq:per_sublattice_norm}
\end{equation}
for each $\mu = 1,\ldots, n_s$.  
\par 
Any spin configuration gives a correlation function satisfying both $\mathcal{G} \succeq 0$ and Eq.~\eqref{eq:per_sublattice_norm}. Therefore, minimizing the energy density Eq.~\eqref{eq:energy_multisub} subject to these two constraints gives a lower bound $\varepsilon_{\cP,1} \leq \varepsilon^*$ on the true ground state energy density $\varepsilon^*$. The resulting optimization problem is a semidefinite program, a class of convex optimizations that can be solved efficiently~\cite{lasserre2009moments}. The subscript $1$ on $\varepsilon_{\cP,1}$ indicates the $r=1$ level of a hierarchy whose higher levels we construct in Sec.~\ref{sec:spin_sdp}. Enlarging the patch $\cP$ adds more constraints and hence monotonically improves $\varepsilon_{\cP,1}$. However, it is necessary to increase the hierarchy level in order to guarantee convergence in general. 
 
\par

\subsection{Relationship to Luttinger-Tisza method}
\label{sec:LT}

The Luttinger--Tisza (LT) method~\cite{luttinger1946theory,lyons1960method} is an analytical approach for finding ground states of quadratic Heisenberg Hamiltonians by Fourier diagonalization by relaxing the spin normalization constraint. As a brief review, consider spins on a Bravais lattice with energy
\begin{equation}
    E[\{\vec S_i\}] = \sum_{ij} J(\vec r_i - \vec r_j)\, \vec S_i \cdot \vec S_j
\end{equation}
subject to the per-site constraints $|\vec S_i|^2 = 1$. LT replaces these with the single global constraint $\sum_i |\vec S_i|^2 = N$, in which case the energy can be written in Fourier space as
\begin{equation}
    E = \sum_{\k} J(\k)\, |\vec S(\k)|^2,
\end{equation}
subject to $\sum_{\vec k} |\vec S(\k)|^2=N$ and a minimum is obtained by placing all weight at the wavevector $\k^* = \arg\min_\k J(\k)$, giving the energy density $\varepsilon_{\mathrm{LT}} = \min_\k J(\k)$.  
\par 
We now show that the Luttinger--Tisza method coincides with the SDP of Sec.~\ref{sec:quadratic} for Bravais lattices. Consider this SDP with $\mathcal G \succeq 0$ enforced on the full infinite lattice rather than a finite patch. On a Bravais lattice, $\mathcal G_{\R,\R'} = g(\R' - \R)$ depends only on the displacement, and we define its Fourier transform $\hat g(\k) \coloneqq \sum_\R g(\R)\, e^{-i\k\cdot\R}$. By Bochner's theorem~\cite{rudin1990fourier}, the positive-semidefiniteness condition $\mathcal G \succeq 0$ is equivalent to $\hat g(\k) \geq 0$ for all $\k$ in the Brillouin zone. The normalization $g(\vec 0) = 1$ becomes $\int_{\mathrm{BZ}} \hat g(\k)\, d^d k / (2\pi)^d = 1$. The SDP thus reduces to minimizing $\int_{\mathrm{BZ}} J(\k)\, \hat g(\k)\, d^d k / (2\pi)^d$ over probability densities on the Brillouin zone, whose minimum places all weight at $\k^*$ and gives $\varepsilon_{\infty,1} = \min_\k J(\k)$. For Bravais lattices, therefore,
\begin{equation}
    \varepsilon_{\mathrm{LT}} = \varepsilon_{\infty,1} \leq \varepsilon^*.
    \label{eq:LT_vs_SDP}
\end{equation}
For isotropic two-spin models on Bravais lattices, the LT minimum can be attained by a unit-length spiral for $O(n)$ spins with $n\geq 2$, apart from finite-size commensurability effects~\cite{nussinov2001commensurate}. On non-Bravais lattices, LT can still be applied but gives strictly weaker bounds, as we demonstrate numerically in Sec.~\ref{sec:centered_sq}.

\subsection{The bootstrap for general Hamiltonians}
\label{sec:spin_sdp}
\begin{figure}[t]
    \centering
    \includegraphics[width=\columnwidth]{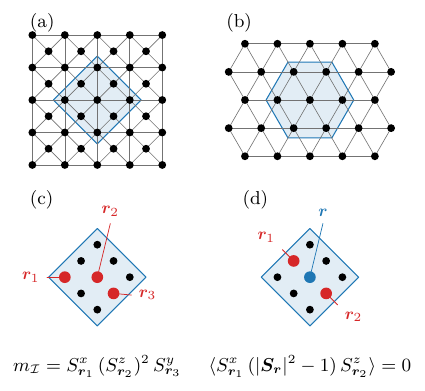}
    \caption{\textit{Lattices, patches, and bootstrap variables.}
  (a),(b) Centered-square and triangular lattices, with shaded regions indicating
  schematic patches $\mathcal{P}$ used to define the SDP.
  (c) An example spin monomial
  $m_{\mathcal{I}} = S^{x}_{\vec r_1}\,(S^{z}_{\vec r_2})^{2}\,S^{y}_{\vec r_3}$
  supported on three sites of the patch;
  its translation-invariant moment
 $y_{[\mathcal{I}]}=\langle m_{\mathcal{I}}\rangle$ is a variable of the SDP.
  (d) The constraint $|\vec S_{\vec r}|^{2}-1 = 0$
  at $\vec r$ generates constraints among moments. Multiplying it by any monomial gives a polynomial whose expectation vanishes,
  e.g.\ $\langle S^{x}_{\vec r_1}(|\vec S_{\vec r}|^{2}-1)\,S^{z}_{\vec r_2}\rangle = 0$.}
    \label{fig:lattices}
\end{figure}

The bootstrap approach presented in Sec.~\ref{sec:quadratic} is the first level ($r = 1$) of a hierarchy that handles arbitrary polynomial Hamiltonians. We now formulate the general construction.  

A translation-invariant (TI) Hamiltonian is specified by its per-site energy density $\mathcal{E}$, a polynomial in the spin variables at finitely many sites near the origin. The ground state energy density is
\begin{equation}
    \varepsilon^* = \inf_{\mu \in M_{\mathrm{TI}}} \langle \mathcal{E}\rangle_\mu,
    \label{eq:gs_energy}
\end{equation}
where $M_{\mathrm{TI}}$ denotes the set of translation-invariant (TI) states. A translation-invariant state is, formally, a translation-invariant probability distribution over spin configurations in the infinite lattice; we formulate this precisely in Appendix~\ref{app:convergence}. From now on, any expectation value $\langle \cdot \rangle$ will denote an expectation value in a translation-invariant state. This restriction does not sacrifice generality, since any spin configuration has the same energy density as its translation average. For correlation functions, the TI average may differ from values in a specific symmetry-broken ground state, and the bounds we derive apply to the symmetry-averaged correlation functions.

We work in terms of spin monomials. A \textit{spin monomial} is a product $m_\mathcal{I} = S^{a_1}_{\r_1} \cdots S^{a_k}_{\r_k}$, where each $\r_j$ specifies a site and sublattice and each $a_j\in \{x,y,z\}$. The number $k = |\mathcal{I}|$ is the \textit{degree} of the monomial, and by convention $m_\emptyset = 1$. Translation invariance implies that $\langle m_\mathcal{I}\rangle$ depends only on the relative positions of the sites; we write $[\mathcal{I}]$ for the equivalence class of $\mathcal{I}$ under uniform Bravais translations and define \begin{equation}
    y_{[\mathcal{I}]} \coloneqq \langle m_{[\mathcal{I}]}\rangle
\end{equation} to be the associated expectation value. Recall that, here, the expectation value is integration against a distribution over the space of spin configurations (see Appendix \ref{app:convergence} for a precise definition). 
\par
The $r$-th level of the SDP hierarchy takes as variables translation-invariant expectation values $y_{[\mathcal{I}]}$ of all monomials $\mathcal{I}$ of degree at most $r$. In practice we fix a finite patch $\mathcal{P}$ and retain only those monomials whose sites all lie in $\mathcal{P}$. This gives a finite index set
\begin{equation}
    B_{\mathcal{P}, r} \coloneqq \{\mathcal{I} : \r_1,\ldots,\r_k \in \mathcal{P},\ |\mathcal{I}| \leq r\}
\end{equation}
and a finite set of variables $y_{[\mathcal{I}]}$.

Next, we define the \textit{moment matrix} $L$, indexed by pairs of monomials in $B_{\mathcal{P}, r}$, with entries
\begin{equation}
    L_{\mathcal{I}, \mathcal{I}'} = \langle m_{[\mathcal{I}]}\, m_{[\mathcal{I}']}\rangle,
    \label{eq:moment_matrix_general}
\end{equation}
which is expressible in terms of the variables $y_{[\mathcal{J}]}$ since $m_\mathcal{I} m_{\mathcal{I}'}$ is itself a monomial. Positive semidefiniteness of $L$ follows from the observation that for any real coefficients $\{c_\mathcal{I}\}$, $\sum c_\mathcal{I} c_{\mathcal{I}'} L_{\mathcal{I}, \mathcal{I}'} = \langle (\sum c_\mathcal{I} m_\mathcal{I})^2\rangle \geq 0$.

We impose the unit-length constraints $\sum_a (S^a_\r)^2 = 1$ as follows. For each site $\r \in \cP$ we define the \textit{localizing matrix}
\begin{equation}
[C_{\vec r}]_{\mathcal{I},\mathcal{I}'} = \left \langle m_{\mathcal{I}}\left(
    \sum_{a} (S_\r^a)^2 -1\right)\, m_{\mathcal I'}\right\rangle 
    \label{eq:localizing}
\end{equation}
with $\mathcal I, \mathcal I' \in B_{\cP, r-1}$, and we impose $C_{\vec r} = 0$. This enforces linear relations among the variables $y_{[\mathcal J]}$. At $r = 1$ (i.e. for at most quadratic monomials), this reduces to the normalization Eq.~\eqref{eq:per_sublattice_norm} of Sec.~\ref{sec:quadratic}.  
\par

Altogether, the SDP at level $r$ is
\begin{equation}
\begin{aligned}
    \varepsilon_{\cP, r} \;=\; \min_{y}\;\; &  \langle \mathcal{E}\rangle \\
\text{s.t.}\;\; & L \succeq 0,\ C_\r = 0\ \forall\, \r \in \cP.
\end{aligned}
\label{eq:tl_sdp}
\end{equation}
As an additional constraint, we require $L_{\emptyset \emptyset} = 1$, which states that the expectation value of the empty monomial is $1$. The minimization involves the variables $\{y_{[\mathcal I]} : \mathcal I \in B_{\cP, 2r}\}$. Enlarging the patch or raising the hierarchy level adds constraints and produces monotonically improving lower bounds:
\begin{equation}
    \varepsilon_{\cP, r} \leq \varepsilon_{\cP', r'} \leq \varepsilon^* \quad \text{for } \cP \subseteq \cP',\ r \leq r'.
    \label{eq:monotone}
\end{equation}
 
\par
At this stage, we have constructed a hierarchy of finite-sized semidefinite programs, which can be efficiently and exactly solved by SDP solvers, whose solutions are lower bounds on the true ground state energy density $\varepsilon^*$. A natural and important question is whether these lower bounds converge to $\varepsilon^*$ as the patch size and hierarchy level grow. Due to the fact that we are dealing with an infinite lattice and infinite series of optimizations, answering this question requires some care. Thankfully, we can use standard mathematical tools in the theory of the moment problem and measure theory to answer this in the affirmative: in Appendix ~\ref{app:convergence}, we provide a full proof that for a sequence $(\cP_n, r_n) \to (\mathbb Z^d, \infty)$, we have 
\begin{equation}
    \lim_{n\to \infty} \varepsilon_{\cP_n,r_n} = \varepsilon^*.
\end{equation}
Here, we explain the main idea of the proof for the interested reader. We make essential use of the fact that all moments $y_{[\mathcal{I}]}$ output by the SDP programs live on a compact space, namely $[-1,1]^{\mathcal{C}}$ where $\mathcal{C} = \{[\mathcal{I}]\}$ is the set of monomial equivalence classes. Owing to compactness, the sequence of moments output by the SDP has a subsequence converging to a limit point $y^*_{[\mathcal{I}]}$. The crux is to show that this limiting set of moments arises from a genuine translation-invariant probability distribution over spin configurations on the infinite lattice, for which we use Kolmogorov's extension theorem. The energy density of this probability distribution equals the limit of the $\varepsilon_{\cP_n,r_n}$ and upper bounds $\varepsilon^*$, forcing it to equal $\varepsilon^*$.

\subsection{Bounds on observables, upper bounds on energy, and SDP spin textures}
\label{sec:observables}

So far, we have focused on bootstrapping lower bounds on the energy density. In this section, we extend this in two directions: to bounds on observables and upper bounds on the energy density. We also explain how one may obtain an approximate ground state spin texture directly from the solution of the SDP. 
\par

First, a generalization of the bootstrap method yields rigorous bounds on arbitrary polynomial observables with finite support. Let $\mathcal{O}$ be any observable written as a polynomial in the spin variables with support within $\cP$ and of degree at most $2r$. Then $\langle \mathcal{O}\rangle$ is a linear function of the moments $y_{[\mathcal I]}$, and all of its constituent monomials appear among the variables of the SDP. Minimizing or maximizing this linear function over the SDP feasible set of Eq.~\eqref{eq:tl_sdp} therefore yields rigorous lower and upper bounds on the ground state expectation value of $\mathcal{O}$. As a concrete example, we will be interested in two-point observables $\mathcal{O} = \vec S_{\vec 0, \mu}\cdot \vec S_{\R, \nu}$.
 \par
We should emphasize that these bounds apply to translation-invariant states, so in the case of translation symmetry breaking, we constrain the symmetry-averaged correlation function rather than the correlation function in a specific symmetry-broken ground state. 
 \par
In practice, we often obtain tighter bounds by combining the observable SDP with an energy constraint. Let $\varepsilon_\mathrm{ub}$ be an upper bound on the ground state energy density obtained from any variational procedure. We then add the constraint $\langle \mathcal{E}\rangle \leq \varepsilon_\mathrm{ub}$ to the SDP, giving
\begin{equation}
\begin{aligned}
    \mathcal{O}_{\min/\max} \;=\; \min/\max\;\; & \langle \mathcal{O}\rangle \\
    \text{s.t.}\;\; & \text{Eq.~\eqref{eq:tl_sdp}},\quad \langle \mathcal{E}\rangle \leq \varepsilon_\mathrm{ub}.
\end{aligned}
\label{eq:observable_sdp}
\end{equation}
These constraints restrict the feasible set to physical moment sequences with energy density compatible with an (optional) upper bound. These observable bounds are rigorous in the thermodynamic limit, and hence can be used in principle to rule out candidate ground state orders or confirm numerical results from complementary methods.  

Next, we discuss how upper bounds on the energy density may be obtained. First, we note that there are standard variational techniques such as classical Monte Carlo which are broadly applicable but can struggle in frustrated systems. In this work, for later results, we consider iterative one-site minimization on a large periodic cluster. Starting from a random initial spin configuration, we update each spin $\vec S_i$ to the unit vector that minimizes the energy with all other spins held fixed, and sweep over the cluster until convergence. For quadratic Hamiltonians, this update aligns $\vec S_i$ antiparallel to its local field $\vec h_i = \sum_j J_{ij}\, \vec S_j$. For non-quadratic Hamiltonians, the update instead minimizes the local energy numerically over the unit sphere. We repeat the procedure from many random initial configurations and retain the lowest energy found, which we denote $\varepsilon_\mathrm{ub}$. Because any normalized, translation-invariant spin configuration gives an energy density that is at least $\varepsilon^*$, we have $\varepsilon_\mathrm{ub} \geq \varepsilon^*$. This procedure works on any lattice and for any polynomial Hamiltonian, but its success depends on the energy landscape being navigable by single-site sweeps from random starting points. In strongly frustrated systems, where there are many competing minima, this procedure can fail to obtain the true ground state energy.  

A second procedure for obtaining an upper bound extracts a candidate spin configuration directly from the SDP solution. To explain this method, we first explain how a spin texture can be extracted from the (approximate) two-point correlation function $\mathcal{G}_{ij} = \langle \vec S_i \cdot \vec S_j\rangle$ provided by the SDP. We observe that for any true spin configuration $\{\vec S_i\}$, the matrix $\vec S_i \cdot \vec S_j = S^x_i S^x_j + S^y_i S^y_j + S^z_i S^z_j$ is the sum of three rank-one outer products, and hence has a rank at most three, the number of spin components. Therefore, if we were given only the matrix $\vec S_i \cdot \vec S_j$, we could generically extract a consistent spin configuration $\{\vec S_i\}$ as follows: eigendecompose 
\begin{equation}
\vec S_i\cdot \vec S_j = \lambda_1 v^1_i v^1_j + \lambda_2 v^2_i v^2_j + \lambda_3 v^3_i v^3_j,
\end{equation}
where $\vec v^a \cdot \vec v^b = \delta^{ab}$, and form a trial texture $\vec S_i = \sqrt{\lambda_1} v_i^1 \hat{\vec x} + \sqrt{\lambda_2} v_i^2 \hat{\vec y} + \sqrt{\lambda_3}v_i^3\hat{\vec z}$. This spin configuration has the correct two-point correlations and norm as long as $\vec S_i\cdot \vec S_i = 1$ for all $i$.\par

To extract an approximately consistent spin configuration from the matrix $\mathcal{G}$ output by the SDP, we can proceed similarly. In particular, we eigendecompose 
\begin{equation}
\mathcal{G}_{ij} = \sum_n \lambda_n\,  v_i^n v_j^n,
\end{equation}
with eigenvalues ordered $\lambda_1 \geq \lambda_2 \geq \cdots \geq 0$, retain the three leading eigenvalues, and set
\begin{equation}
 \tilde{\vec S}_i = \sum_{n=1}^3 \sqrt{\lambda_n}\, v^n_i\, \hat{\vec e}_n,
    \label{eq:rank3_texture}
\end{equation}
where $\hat{\vec e}_{1,2,3} = \hat{\vec x},\hat{\vec y},\hat{\vec z}$. We then normalize $\hat{\vec S}_i = \tilde{\vec S}_i / |\tilde{\vec S}_i|$. This approximation is not guaranteed to reproduce $\mathcal{G}_{ij}$. However, it is a valid variational spin texture and often reproduces $\mathcal{G}_{ij}$ to good accuracy (an important exception being when there is a large ground state degeneracy).

\blue{The lower and upper bounds are both rigorous, but obtained independently: the former from the SDP, and the latter from a variational spin configuration.} Together, the SDP lower bound and the variational upper bound from either procedure give the two-sided bracket
\begin{equation}
    \varepsilon_{\cP, r} \leq \varepsilon^* \leq \varepsilon_\mathrm{ub}
    \label{eq:two_sided}
\end{equation}
on the ground state energy density. This bracket holds directly in the thermodynamic limit without finite-size extrapolation, and its width quantifies the precision of the method on a given problem. In the examples considered here, we find bracket widths ranging from $10^{-5}$ to $10^{-6}$ per site in units of the coupling.   

\section{Results}
\label{sec:results}

We now numerically demonstrate our bootstrap method on two classical frustrated spin models in two dimensions, whose lattices are shown in Fig.~\ref{fig:lattices}. The first is a quadratic Heisenberg model on a centered-square lattice, chosen as an example of a frustrated model with inequivalent sites and non-uniform coordination. The second model is a quartic Heisenberg model on a triangular lattice, also known as the bilinear--biquadratic model, which is highly frustrated and known to possess an extensive ground state degeneracy at intermediate couplings. To obtain these numerical results, the SDPs were formulated in JuMP~\cite{Lubin2023} and solved with \texttt{MOSEK}~\cite{mosek}.

\subsection{Centered-square lattice}
\label{sec:centered_sq}

The centered-square lattice is shown in Fig.~\ref{fig:lattices}a. This lattice has inequivalent sites with non-uniform coordinations: corner sites have coordination $z = 8$, while center sites have $z = 4$. This is a simple non-Bravais lattice where Luttinger--Tisza faces challenges, as we show.   
\par 
In particular, the Hamiltonian is assumed to include only nearest-neighbor antiferromagnetic Heisenberg interactions,
\begin{equation}
    H = J_1 \sum_{\langle ij\rangle_{\mathsf{+}}} \vec S_i \cdot \vec S_j + J_2 \sum_{\langle ij\rangle_{\mathsf{\times}}} \vec S_i \cdot \vec S_j,    \label{eq:centered_square_H}
\end{equation}
where $\langle ij\rangle_{\mathsf{+}}$ denotes horizontal and vertical nearest-neighbor bonds and $\langle ij\rangle_{\mathsf{\times}}$ denotes diagonal nearest-neighbor bonds. We fix $J_1 = 1$.   
\par 

In Figure~\ref{fig:centered_square}, we show the SDP lower bound, the variational upper bound, and the LT result for varying $J_2/J_1$ for the ground state energy density $\varepsilon$ and nearest-neighbor two-point correlation functions $g_\mathrm{nn} \coloneqq \langle \vec S_0 \cdot \vec S_1\rangle$. The SDP bound and the variational upper bound nearly coincide everywhere. The LT bound is loose for $J_2 / J_1 \lesssim 1.4$, demonstrating the advantage of the SDP method. In Figure~\ref{fig:centered_square}(c) we show that the agreement between the SDP and the variational procedure extends to the spin textures produced by each method, which are visually identical after a global $O(3)$ alignment. Altogether, this demonstrates the advantage of the SDP method over the LT method and complementarity with finite-size variational methods, which the SDP can certify with rigorous, thermodynamic-limit results.

\begin{figure}[t]
    \centering
    \includegraphics[width=\columnwidth]{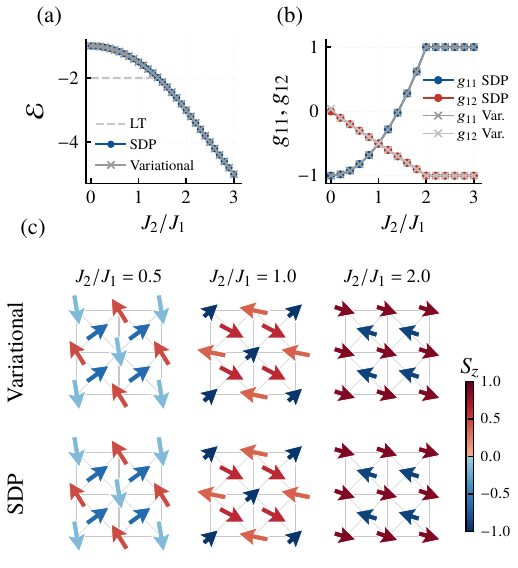}
\caption{\textit{Centered-square lattice model.} Couplings are $J_1$ (horizontal and vertical bonds) and $J_2$ (diagonal bonds). (a) SDP lower bound (blue), variational upper bound (gray crosses), and Luttinger--Tisza (gray dashed) for the ground state energy density $\varepsilon$. (b) Nearest-neighbor correlators $g_{11}$ (corner--corner) and $g_{12}$ (corner--center) from the SDP (blue and red circles) and the variational procedure (gray crosses). (c) Spin textures from the variational procedure (top) and rank-3 rounding of the SDP solution (bottom) at three values of $J_2/J_1$. Textures are aligned by a global $O(3)$ rotation.}
    \label{fig:centered_square}
\end{figure}
\begin{figure}[t]
    \centering
    \includegraphics[width=\columnwidth]{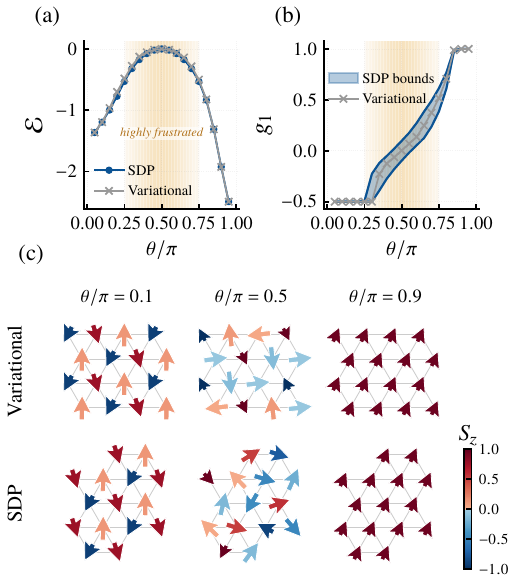}
\caption{\textit{Bilinear--biquadratic model on the triangular lattice.} (a) SDP lower bound (blue) and variational upper bound (gray crosses) on the ground state energy density $\varepsilon$ vs.\ $\theta/\pi$. (b) Two-sided bounds (blue band) on the nearest-neighbor correlator $g_\mathrm{nn}$. \blue{The variational results (gray crosses) lie within these bounds to numerical accuracy.} (c) Spin textures from the variational procedure (top) and rank-3 rounding of the SDP solution (bottom) at $\theta/\pi = 0.1$ ($120^\circ$ antiferromagnetic), $0.5$ (frustrated), and $0.9$ (ferromagnetic). Textures are aligned by a global $O(3)$ rotation.}
\label{fig:blbq}
\end{figure}

\subsection{Quartic triangular antiferromagnet}
\label{sec:bbq}

We now turn to the bilinear-biquadratic (BLBQ) model on the triangular lattice~\cite{lauchli2006quadrupolar,tsunetsugu2006nematic}, a highly frustrated quartic Heisenberg model. Because the Hamiltonian is quartic, Luttinger--Tisza does not apply. The Hamiltonian is
\begin{equation}
    H = \frac{1}{2}\sum_{\langle ij\rangle}
    \bigl[\cos\theta\, \vec{S}_i \cdot \vec{S}_j
    + \sin\theta\, (\vec{S}_i \cdot \vec{S}_j)^2 \bigr],
\label{eq:blbq_hamiltonian}
\end{equation}
where the sum is over nearest-neighbors on the triangular lattice. The model has a well-characterized phase diagram: as $\theta$ sweeps from $0$ to $\pi$, there is a $120^\circ$ antiferromagnet, a highly frustrated window with extensive ground state degeneracy, and a ferromagnet.  
\par
In Figure~\ref{fig:blbq}(a) we show the SDP lower bound and the variational upper bound on $\varepsilon$ across the phase diagram. For the SDP, we use an $r = 2$ moment matrix on a four-site patch jointly with an $r = 1$ two-point moment matrix on a 25-site patch. The $r = 2$ patch is sufficient for tight energy bounds, while the $r = 1$ patch is helpful for extracting the spin texture. 
\blue{The corresponding patch sizes, hierarchy levels, solver settings, and run times are summarized in Table~\ref{tab:sdp}.} The two bounds nearly coincide in the $120^\circ$ and ferromagnetic phases ($\theta \approx 0$ and $\pi$, respectively). Figure~\ref{fig:blbq}(b) shows two-sided bounds on the nearest-neighbor spin correlation function $g_\mathrm{nn} \coloneqq \langle \vec S_0 \cdot \vec S_1\rangle$, obtained using Eq.~\eqref{eq:observable_sdp} and a variational upper bound $\varepsilon \leq \varepsilon_\mathrm{ub}$. In the $120^\circ$ phase ($\theta/\pi \lesssim 0.25$), $g_\mathrm{nn}=-1/2$ and in the ferromagnetic phase ($\theta/\pi \gtrsim 0.85$), $g_\mathrm{nn}$ approaches $1$, as expected. In the intermediate region where the quartic interaction dominates and ground state degeneracy is extensive, the bracket is much wider, reflecting the intrinsic difficulty of this regime. However, the SDP results correctly certify the variational results. In Figure~\ref{fig:blbq}(c) we compare variational and SDP-rounded spin textures at three representative couplings. In the $120^\circ$ and ferromagnetic phases, the SDP texture matches the variational one after $O(3)$ alignment. In the highly frustrated regime at $\theta/\pi = 0.5$, the two-point moment matrix has rank 25 (the maximum), consistent with a large degenerate ground state manifold. In this case, the rank-3 projection selects an arbitrary representative from this manifold.   
\par
\blue{To further demonstrate the scope of the method, we also study the triangular-lattice Heisenberg--Kitaev model in App.~\ref{app:hk}~\cite{becker2015spinorbit,rousochatzakis2016kitaev}. Its bond-dependent interactions break the global $O(3)$ symmetry down to a discrete subgroup, but require no modification of the bootstrap. We again find that enlarging the spatial patch systematically improves the bounds on both the energy and spin correlations.}

\section{Discussion}
\label{sec:discussion}

We have introduced a bootstrap method for classical frustrated magnetism based on classical results in polynomial optimization~\cite{lasserre2001global}, producing rigorous bounds on ground state energy densities and arbitrary spin correlation functions directly in the thermodynamic limit. The hierarchy is indexed by a patch $\cP$ and a level $r$, converges monotonically as $\cP$ and $r$ grow to the true ground state results, and reduces to the Luttinger--Tisza method at lowest level on Bravais lattices, while extending naturally to non-Bravais lattices and non-quadratic Hamiltonians. Numerical results on two-dimensional frustrated models indicate excellent performance in practice, showing close agreement with variational results while providing previously absent guarantees. Moreover, it is straightforward to extract explicit candidate ground-state spin textures and each SDP solve is computationally efficient. Hence, we believe this method can be a versatile general-purpose tool for classical frustrated magnetism.
\par 

\blue{One distinction from quantum bootstrap methods concerns Hamiltonian-dependent positivity constraints such as $\langle O^\dagger[H,O]\rangle\geq0$~\cite{fawzi2024certified,cho2026coarse}. These rely on the noncommutativity of quantum operators and have no direct counterpart in the present classical formulation. Classical ground states do satisfy local stability conditions, but these also hold at metastable minima and therefore do not by themselves certify the ground state.}
\par
\blue{The bootstrap can nevertheless provide information about ordering and degeneracy. Bounds on correlators constrain their possible variation among ground states, while Fourier transforming the two-point correlations identifies candidate ordering wavevectors, with resolution limited by the spatial patch. As with Luttinger--Tisza, these signatures can guide the search for ground states but do not by themselves establish extensive degeneracy or certify a particular spin configuration.}
\par
\blue{Finally, anisotropic exchange, cubic terms, and cyclic interactions can be incorporated without changing the bootstrap construction. The main additional cost is computational, as breaking $O(3)$ symmetry increases the size of the moment matrix, with the $r=1$ two-point block growing from $|\cP|$ to $3|\cP|$, while higher-degree interactions may require larger $r$. We illustrate this extension with the Heisenberg--Kitaev model in App.~\ref{app:hk}.}
\par 

Several extensions would be natural. One example is relaxing the requirement for translation invariant Hamiltonians to allow for quenched disorder, thereby probing ground state energies and correlation functions in classical spin glasses, a longstanding problem. It would also be interesting to consider a semiclassical bootstrap method in between our method and the fully quantum many-body bootstrap, which works with noncommutative operators and the full density matrix, as would be relevant to the leading quantum corrections at large spin. Finally, it would be interesting to investigate whether bootstrap approaches can be useful in studying classical dynamical processes~\cite{fantuzzi2016bounds,doering2020optimal,cho2025bootstrapping}, including Landau–Lifshitz–Gilbert dynamics in magnetism. We leave these directions to future work.
\par 
\blue{\paragraph*{Code availability.} Code for the implementation of the bootstrap and reproduction of results is available at \href{https://github.com/nisargapaul/ClassicalSpinBootstrap}{\texttt{\blue{github.com/nisargapaul/ClassicalSpinBootstrap}}}.}
\par 

\paragraph*{Acknowledgments.} We thank Roderich Moessner and John Chalker for helpful comments and Yugo Onishi and Mayur Sharma for related discussions. We thank David Simmons-Duffin for pointing out helpful references. NP acknowledges support from the Walter Burke Institute for Theoretical Physics at Caltech. NP acknowledges Claude Code $\leq$ Opus 5.5 and Codex GPT $\leq 6$ for coding assistance and curating the Github repo. GR is grateful for support from the Simons Foundation,
the Institute for Quantum Information and Matter,
an NSF Physics Frontiers Center (PHY-2317110), and
the AFOSR MURI program, under agreement number
FA9550-22-1-0339.

\bibliographystyle{apsrev4-2}
\bibliography{bib}

\appendix
\section{The Lasserre hierarchy}
\label{app:lasserre}

The SDP bootstrap of Sec.~\ref{sec:spin_sdp} is an instance of the Lasserre hierarchy~\cite{lasserre2001global,lasserre2009moments}, a general framework for polynomial optimization via semidefinite relaxations. We review the Lasserre hierarchy briefly here. Suppose we wish to find the global minimum $p^\star$ of a real polynomial $p(x): \mathbb{R}^n \to \mathbb{R}$:
\begin{equation}
    p^\star = \min_{x\in \mathbb{R}^n} p(x).
\end{equation}
Directly minimizing $p(x)$ is generally computationally hard. One approach, scaling poorly in general, is to enumerate all local extrema $\nabla p(x) = 0$ and search for the global minimum. Lasserre's approach is to lift the problem from the space of points $x$ to the space of probability distributions $\mathcal{P}(\mathbb{R}^n)$. We rewrite the optimization in an equivalent form~\cite{lasserre2001global}:
\begin{equation}\label{eq:muproblem}
    p^\star = \min_{\mu \in \mathcal{P}(\mathbb{R}^n)} \int p(x) \, d\mu(x).
\end{equation}
The equivalence holds because the optimal distribution is simply a Dirac delta function $\delta(x-x^\star)$ concentrated at the global minimum, assuming the minimum is finite and attained. The advantages of this formulation are convexity and linearity. We may encode the distribution $\mu$ by its sequence of moments $y_\alpha = \int x^\alpha\,d\mu(x)$, where $x^\alpha \coloneqq x_1^{\alpha_1}\cdots x_n^{\alpha_n}$. The highly nonlinear objective function becomes a linear functional of these moments:
\begin{equation}
L_y(p) \coloneqq \sum_\alpha c_\alpha y_\alpha = \int p(x)\,d\mu(x).
\end{equation}
The trade-off is that we must characterize which sequences $y_\alpha$ correspond to valid distributions. For example, no probability distribution can have $y_2 = \int x^2 d\mu(x) <0$. A necessary (and sufficient for $n=1$) condition is \emph{positivity}: for any polynomial $g(x)$, the expectation of its square must be non-negative:
\begin{equation}
    \int g(x)^2 d\mu(x) \ge 0.
\end{equation}
At a finite relaxation level $d$, we enforce this condition for all test polynomials $g(x)$ of degree $\le d$. Let $\mathcal{M}_d(x)$ denote the vector of all monomials $x^\alpha$ with total degree $|\alpha|:=\sum_k \alpha_k\le d$. Writing $g(x) = c^\top \mathcal{M}_d(x)$ in a monomial basis, the condition $L_y(g^2) \ge 0$ becomes the requirement that the \emph{moment matrix} $M_d(y)$ is positive semidefinite:
\begin{equation}
    M_d(y)_{\alpha,\beta} := y_{\alpha+\beta} \succeq 0.
\end{equation}
This yields Lasserre's hierarchy of semidefinite programs~\cite{lasserre2001global}:
\begin{equation}\label{eq:lasserrePrimal}
\bar p_d = \min_y L_y(p) \quad \text{s.t.} \quad y_0=1,\ M_d(y)\succeq 0.
\end{equation}
Any valid distribution satisfies these moment constraints, so $\bar p_d \le p^\star$. As an explicit example, consider the task of minimizing the polynomial $p(x_1,x_2) = x_1^2 x_2^2+x_1^2+4 x_1 x_2 -2 x_1+x_2^2+2 x_2+6$. At degree $d=2$, the optimization variables are the moments $y_{\alpha}$ up to degree 4. The objective function is linear: $y_{2,2}+y_{2,0}+4y_{1,1}-2y_{1,0}+y_{0,2}+2y_{0,1} +6$. The constraint $M_2(y) \succeq 0$ explicitly requires the following $6 \times 6$ matrix to be positive semidefinite:
\begin{equation}
    \begin{pmatrix}
        1 & y_{1,0} & y_{0,1} & y_{2,0} & y_{1,1} & y_{0,2} \\
        y_{1,0} & y_{2,0} & y_{1,1} & y_{3,0} & y_{2,1} & y_{1,2} \\
        y_{0,1} & y_{1,1} & y_{0,2} & y_{2,1} & y_{1,2} & y_{0,3} \\
        y_{2,0} & y_{3,0} & y_{2,1} & y_{4,0} & y_{3,1} & y_{2,2} \\
        y_{1,1} & y_{2,1} & y_{1,2} & y_{3,1} & y_{2,2} & y_{1,3} \\
        y_{0,2} & y_{1,2} & y_{0,3} & y_{2,2} & y_{1,3} & y_{0,4}
    \end{pmatrix}\succeq 0.
\end{equation}
Solving this SDP yields a rigorous lower bound on the global minimum. As $d$ increases, the set of test polynomials grows, tightening the bound. Under a further mild boundedness assumption, the sequence converges to the true minimum, $\lim_{d\to\infty} \bar p_d = p^\star$~\cite{schmudgen2017moment}.

\section{Convergence of the hierarchy}
\label{app:convergence}

The bootstrap approach that we have introduced naturally produces a hierarchy of lower bounds on the true thermodynamic-limit ground state energy density for translation-invariant Hamiltonians. The lower bounds get tighter with computational effort. It is desirable, from both practical and formal perspectives, to prove that this procedure \textit{converges}: that is, to show that as the spatial patch $\cP$ grows and the hierarchy level $r$ increases, the SDP lower bounds $\varepsilon_{\cP,r}$ converge to the \textit{true} ground state energy density $\varepsilon^*$.\par 

In this Appendix, we prove this convergence. Part of the difficulty of this proof is working with translation-invariant probability measures in an infinite-dimensional configuration space, so some effort is necessary to precisely define this notion. The finite-dimensional SDPs used in practice are finite projections of the full constraint set to monomials supported in $\cP$ and degree at most $r$; we'll define an object $\mathcal{F}_{\cP,r}$ to describe this. After this is defined, the proof proceeds in three main steps: (1) showing that any sequence of bootstrap optimizers has a convergent subsequence, (2) a limit point of such a subsequence defines a genuine translation-invariant probability measure on the infinite spin configuration space, and (3) the energy density of this measure equals the true ground state energy density. Our proof is aimed to be readable by a physics audience familiar with elementary topology and measure theory. \par 

We will be interested in the spin configuration space
\begin{equation}
    X \coloneqq (S^2)^\Gamma,
    \qquad
    \Gamma \coloneqq \mathbb{Z}^d \times \{1,\dots,n_s\}
\end{equation}
where $S^2$ is the unit $2$-sphere, $n_s$ is the number of sublattices, and $\Gamma$ is the lattice. We equip $X$ with the product topology, where each copy of $S^2$ carries its usual topology as a subspace of $\mathbb{R}^3$. This means that a sequence $s^{(n)} \in X$ converges to $s \in X$ if and only if, for each fixed site $i \in \Gamma$, $s^{(n)}_i \to s_i$ in $S^2$. Moreover, since $S^2$ is compact and $\Gamma$ is countable, $X$ is compact by Tychonoff's theorem.  

We will need to define probability measures on $X$. For this, the structure of a topology (which defines the collection of open sets) is not sufficient, and we need the structure of a $\sigma$-algebra: a collection of subsets closed under complement, countable unions, and countable intersections. The canonical choice is to take the \textit{Borel} $\sigma$-algebra $\mathcal{B}(X)$, defined as the smallest $\sigma$-algebra containing all open sets of $X$. \par 

A probability measure is now defined to be a map $\mu: \mathcal{B}(X)\to [0,1]$ such that $\mu(\emptyset) = 0$, $\mu(X) = 1$, and for any countable collection of pairwise disjoint Borel sets $A_1,A_2,\ldots \in \mathcal{B}(X)$, we have
\begin{equation}
\mu\left(\bigcup_{i=1}^{\infty}A_i\right) = \sum_{i=1}^{\infty} \mu(A_i).
\end{equation}
We now define a \textbf{state} to be a translation-invariant probability measure on $X$. That is, $\mu$ satisfies the above conditions as well as 
\begin{equation}
\mu(T_{\vec t}A) = \mu(A)\qquad \forall  \vec t \in \mathbb{Z}^d
\end{equation}
where $T_{\vec t}$ is a translation. We write $M_{\mathrm{TI}}(X)$ for the set of all states. 
\par

The ground state energy density is then
\begin{equation}
    \varepsilon^*
    \coloneqq
    \inf_{\mu \in M_{\mathrm{TI}}(X)}
    \int_X \mathcal{E}\, d\mu,
\end{equation}
where $\mathcal{E}:X\to \mathbb{R}$ is the local energy density defined in
Sec.~\ref{sec:spin_sdp}. The integral is understood in the usual
Lebesgue sense with respect to the probability measure $\mu$. It is well defined because $\mathcal{E}$ is a local polynomial observable: it depends on only finitely many spin variables, and is therefore a bounded continuous function on $X$. 
\par 

Next, we need to define the notion of \textit{moments} of the state $\mu$. It will be convenient to introduce a compact notation for spin monomials. Since the spins carry both a lattice-site label and a Cartesian component, we first combine these into a single index set
\begin{equation}
    K \coloneqq \Gamma \times \{x,y,z\}.
\end{equation}
An element of $K$ is therefore a pair $(i,a)$, where $i \in \Gamma$ labels a lattice site (including sublattice index) and $a \in \{x,y,z\}$ labels a spin component. Then a monomial in the spin variables is determined by specifying, for each pair $(i,a)$, how many times the variable $s_i^a$ appears. This is encoded by a function
\begin{equation}
    \alpha : K \to \mathbb{N},
\end{equation}
where $\mathbb{N} = \{0,1,2,\dots\}$. We will only be interested in monomials involving finitely many spin variables, so we require $\alpha$ to be finitely supported, meaning that $\alpha(i,a)\neq 0$ for only finitely many pairs $(i,a)$. We denote the set of all such functions by $\mathbb{N}^{(K)}$. \par

Given such an $\alpha$, we define its total degree by
\begin{equation}
    |\alpha|
    \coloneqq
    \sum_{(i,a)\in K} \alpha(i,a),
\end{equation}
and its spatial support by
\begin{equation}
    \supp(\alpha)
    \coloneqq
    \{i\in \Gamma : \exists a \text{ such that } \alpha(i,a)\neq 0\}.
\end{equation}
Thus $|\alpha|$ counts the total number of spin factors in the monomial, while $\supp(\alpha)$ records which lattice sites appear. \par

The corresponding monomial is
\begin{equation}
    m_\alpha(s)
    \coloneqq
    \prod_{(i,a)\in K} (s_i^a)^{\alpha(i,a)},
\end{equation}
where $s=(s_i)_{i\in\Gamma}\in X$. Because $\alpha$ is finitely supported, this product has only finitely many nontrivial factors and is therefore well defined. We assign the zero multi-index $\alpha=0$ to the constant monomial
\begin{equation}
    m_0 = 1.
\end{equation}
Finally, for each $(i,a)\in K$, we write $e_{i,a}$ for the unit multi-index supported at the single pair $(i,a)$, i.e. the one for which $e_{i,a}(i,a)=1$ and all other entries vanish.

Bravais translations act on the lattice \(\Gamma\) by shifting the Bravais-lattice coordinate while leaving the sublattice index unchanged:
\begin{equation}
   T_{\vec t} \cdot (\R,\mu) = (\R+\vec t,\mu),
\end{equation}
where $\vec t\in \mathbb{Z}^d$. Since a multi-index $\alpha$ records which spin components appear in a monomial, translation acts on multi-indices as well: translating $\alpha$ shifts all of the lattice sites appearing in it by the same vector $\vec t$. Concretely, we define
\begin{equation}
    (T_{\vec t}\alpha)(i,a)\coloneqq \alpha( T_{\vec t}^{-1}\!\cdot i,a).
\end{equation}
Two multi-indices should be regarded as equivalent if they differ only by an overall translation of the lattice, since states assign them the same expectation value. We therefore write $[\alpha]$ for the translation-equivalence class of $\alpha$, and let $\mathcal{C}$ denote the set of all such equivalence classes. We note that the set $\mathcal{C}$ is countable. \par

For any multi-index $\alpha$, the corresponding monomial $m_\alpha$ is a bounded function on $X$, so its expectation value with respect to a state $\mu$ is well defined. Translation invariance of $\mu$ implies that this expectation depends only on the translation-equivalence class $[\alpha]$, not on the particular representative. We may therefore define the associated translation-invariant moment by
\begin{equation}
    y^\mu_{[\alpha]} \coloneqq \int_X m_\alpha\, d\mu.
\end{equation}
These moments are the basic variables of the bootstrap. \par

Because each spin component takes values in the interval $[-1,1]$, every monomial $m_\alpha$ also takes values in $[-1,1]$. It follows immediately that
\begin{equation}
    |y^\mu_{[\alpha]}| \le 1
\end{equation}
for every $\alpha$. Thus every state $\mu$ determines a point in the product space
\begin{equation}
    Y \coloneqq [-1,1]^{\mathcal{C}}.
\end{equation}
Since $\mathcal{C}$ is countable, $Y$ is compact. This compactness will be important later for our convergence proof.

\par

The local energy density $\mathcal{E}$ is, by construction, a polynomial involving only finitely many spin variables. We may therefore write it as a finite linear combination of monomials,
\begin{equation}
    \mathcal{E} = \sum_{j=1}^M a_j\, m_{\alpha_j},
\end{equation}
for some coefficients $a_j$ and multi-indices $\alpha_j$. The energy density is determined by the moments as
\begin{equation}
    \langle \mathcal{E} \rangle_y \coloneqq \sum_{j=1}^M a_j\, y_{[\alpha_j]}.
\end{equation}
This is a finite linear combination of moment coordinates, and hence a (trivially) continuous function on $Y$. \par 

We now define the finite-dimensional SDP relaxation. Fix a finite patch $\cP \subset \Gamma$ and an integer $r \ge 1$. We let
\begin{equation}
    B_{\cP,r}
    \coloneqq
    \{\alpha\in \mathbb{N}^{(K)} : \supp(\alpha)\subset \cP,\ |\alpha|\le r\}
\end{equation}
denote the set of all multi-indices whose associated monomials are supported inside the patch $\cP$ and have degree at most $r$. In other words, $B_{\cP,r}$ collects exactly the monomials that we retain at hierarchy level $r$ on the patch $\cP$. We will always assume that $(\cP,r)$ is chosen large enough that every monomial appearing in $\mathcal{E}$ has some representative supported in $\cP$ and degree at most $2r$, so that the energy density is in the span of the monomials we will consider. \par

The corresponding feasible set $\mathcal{F}_{\cP,r}\subset Y$ consists of all putative moment assignments
\begin{equation}
    y=(y_{[\alpha]})_{[\alpha]\in \mathcal{C}}
\end{equation}
satisfying three conditions. \begin{enumerate}
    \item $y_{[0]}=1$;
    \item the \textit{truncated moment matrix}
    \begin{equation}
        M_{\cP,r}(y)
        \coloneqq
        \bigl(y_{[\alpha+\beta]}\bigr)_{\alpha,\beta\in B_{\cP,r}}
    \end{equation}
    is positive semidefinite;
    \item for every site $i\in \cP$ and every
    $\alpha,\beta\in B_{\cP,r-1}$,
    \begin{equation}
        \sum_{a=x,y,z} y_{[\alpha+\beta+2e_{i,a}]}
        =
        y_{[\alpha+\beta]},
        \label{eq:app_localizing}
    \end{equation}
    which encodes $|\vec S_i|^2=1$ at the level of moments.
\end{enumerate}

Together, these conditions define the \textit{feasible set} $\mathcal{F}_{\cP,r}$. We may now define the SDP lower bound at patch $\cP$ and hierarchy level $r$ by
\begin{equation}
    \varepsilon_{\cP,r}
    \coloneqq
    \min_{y\in \mathcal{F}_{\cP,r}} \langle \mathcal{E}\rangle_y.
\end{equation}
This is the quantity produced by the bootstrap. It is a lower bound on the true ground state energy density because every state $\mu$ produces a feasible moment sequence $y^\mu \in \mathcal{F}_{\cP,r}$. Therefore
\begin{equation}
    \varepsilon_{\cP,r}\le \varepsilon^*.
\end{equation}
Moreover, enlarging the patch or increasing the hierarchy level only adds constraints, and therefore can only increase the minimum. In particular, if $\cP\subseteq \cP'$ and $r\le r'$, then
\begin{equation}
    \mathcal{F}_{\cP',r'}\subseteq \mathcal{F}_{\cP,r},
    \qquad
    \varepsilon_{\cP,r}\le \varepsilon_{\cP',r'}\le \varepsilon^*.
    \label{eq:app_monotone}
\end{equation}
Thus the bootstrap produces a monotone hierarchy of lower bounds. The content of the theorem below is that this hierarchy in fact converges to the true value. \par

Before proceeding with the theorem, we will need two lemmas. The first lemma is a result in the theory of the finite-dimensional moment problem addressing when a linear functional from polynomials to the real numbers arises from the integration of polynomials against some measure~\cite{schmudgen2017moment}, adapted to the case at hand. The second lemma is a technical result on consistency of finite-dimensional marginal distributions which will be necessary for application of the Kolmogorov extension theorem. 
\par 

\begin{lemma}[Moment problem]
\label{lem:finite_patch_measure}
Let $y^*\in \bigcap_{\cP,r}\mathcal{F}_{\cP,r}$.
For every finite $\Lambda\subset \Gamma$, there exists a Borel
probability measure $\mu_\Lambda$ on
\begin{equation}
    K_\Lambda \coloneqq (S^2)^\Lambda
\end{equation}
such that
\begin{equation}
    y^*_{[\alpha]} = \int_{K_\Lambda} m_\alpha\, d\mu_\Lambda
\end{equation}
for every monomial $m_\alpha$ supported in $\Lambda$.
\end{lemma}

\begin{proof}
Fix a finite patch $\Lambda\subset \Gamma$. We would like to show that
the restriction of $y^*$ to monomials supported in $\Lambda$ defines a sequence of moments which comes from a genuine
probability measure on $K_\Lambda$. \par

Let $\mathbb{R}[\Lambda]$ denote the polynomial algebra in the (commuting)
variables
\begin{equation}
    \{s_i^a : i\in \Lambda,\ a=x,y,z\}.
\end{equation}
Define a linear functional
\begin{equation}
    L_\Lambda : \mathbb{R}[\Lambda]\to \mathbb{R}
\end{equation}
by declaring, for every multi-index $\alpha$ supported in $\Lambda$,
\begin{equation}
    L_\Lambda(m_\alpha)\coloneqq y^*_{[\alpha]},
\end{equation}
and then extending by linearity. We first claim that 
\begin{equation}
    L_\Lambda(q^2)\ge 0
    \qquad
    \forall\, q\in \mathbb{R}[\Lambda].
\end{equation}
Indeed, choose $\cP\supset \Lambda$ and $r$ large enough that every
monomial appearing in $q$ lies in $B_{\cP,r}$. Writing $ q=\sum_\alpha c_\alpha m_\alpha$, we find
\begin{equation}
    L_\Lambda(q^2)
    =
    \sum_{\alpha,\beta} c_\alpha c_\beta\, y^*_{[\alpha+\beta]}
    =
    \vec c^{\,T} M_{\cP,r}(y^*) \vec c
    \ge 0
\end{equation}
where we used $y^*\in \mathcal{F}_{\cP,r}$ and
$M_{\cP,r}(y^*)\succeq 0$. 

Next, for each site $i\in \Lambda$, define the sphere constraint
polynomial
\begin{equation}
    h_i(s)\coloneqq 1-\sum_{a=x,y,z}(s_i^a)^2.
\end{equation}
We claim that
\begin{equation}
    L_\Lambda(q^2 h_i)=0
    \qquad
    \forall\, q\in \mathbb{R}[\Lambda],\ \forall\, i\in \Lambda.
\end{equation}
Again, choosing $\cP\supset \Lambda$ and $r$ large enough that every
monomial appearing in $q$ lies in $B_{\cP,r-1}$, writing
$q=\sum_\alpha c_\alpha x^\alpha$ and using the unit norm constraints, we find
\begin{align}
    L_\Lambda(q^2 h_i)
    &=
    \sum_{\alpha,\beta} c_\alpha c_\beta
    \left(
        y^*_{[\alpha+\beta]}
        -
        \sum_{a=x,y,z} y^*_{[\alpha+\beta+2e_{i,a}]}
    \right)
    \nonumber\\
    &=0.
\end{align}
We now have that $L_\Lambda$ is a positive linear functional on $\mathbb{R}[\Lambda]$ which annihilates each of the sphere constraints $h_i$ defining the compact algebraic set $K_\Lambda=(S^2)^\Lambda$. In this setting, where the sphere constraints are quadratic and the resulting algebraic set is compact, that  $L_\Lambda$ is represented by a Borel probability measure is a standard result in the theory of the moment problem (e.g. Theorem 12.36 of Ref~\cite{schmudgen2017moment}).

\end{proof}

\begin{lemma}[Consistency of finite-dimensional marginals]
\label{lem:consistency}
If $\Lambda'\subset \Lambda$, then the marginal of $\mu_\Lambda$ on
$K_{\Lambda'}$ coincides with $\mu_{\Lambda'}$.
\end{lemma}

\begin{proof}
Let $\pi_{\Lambda\to\Lambda'} : K_\Lambda \to K_{\Lambda'}$ be the coordinate projection, and let $\nu_{\Lambda'}$ denote the marginal of $\mu_\Lambda$ on $K_{\Lambda'}$. We claim that $\nu_{\Lambda'}=\mu_{\Lambda'}$. \par

For every polynomial $p\in \mathbb{R}[\Lambda']$, we have
\begin{align}
    \int_{K_{\Lambda'}} p\, d\nu_{\Lambda'}
    &=
    \int_{K_\Lambda} (p\circ \pi_{\Lambda\to\Lambda'})\, d\mu_\Lambda
    \nonumber\\
    &=
    \int_{K_{\Lambda'}} p\, d\mu_{\Lambda'}.
\end{align}
Both integrals are equal to the same moments prescribed by $y^*$, since $p$ and $p\circ \pi_{\Lambda\to\Lambda'}$ represent the same polynomial observable restricted to the smaller patch $\Lambda'$. Thus $\nu_{\Lambda'}$ and $\mu_{\Lambda'}$ agree on all polynomial functions. Since polynomials are dense in the continuous functions on the compact space $K_{\Lambda'}$ by the Stone--Weierstrass theorem, the two measures agree on all continuous functions, and hence are equal as Borel probability measures. Therefore $\nu_{\Lambda'}=\mu_{\Lambda'}$.
\end{proof}

We now state and prove the main theorem. 

\begin{theorem}[Asymptotic completeness]
\label{thm:completeness}
Let $(\cP_n,r_n)_{n\ge 1}$ be a sequence satisfying
\begin{gather}
    \cP_1\subset \cP_2\subset \cdots,
    \qquad
    \bigcup_{n\ge 1}\cP_n=\Gamma,
    \nonumber\\
    r_1\le r_2\le \cdots,
    \qquad
    r_n\to \infty.
\end{gather}
Then
\begin{equation}
    \lim_{n\to\infty} \varepsilon_{\cP_n,r_n} = \varepsilon^*,
\end{equation}
where $\varepsilon^*$ is the true ground state energy density. 
\end{theorem}
\begin{proof}
The proof has three steps. First, we extract a limit point of the SDP optimizers using compactness. Second, we show that such a limit point defines a genuine state on the infinite configuration space. Third, we show that the energy density of this state is exactly $\varepsilon^*$. \par

The key idea is the following. A feasible SDP point is a collection of numbers that is intended to represent the moments of some state, but a priori it is only required to satisfy finitely many consistency conditions: positivity of a truncated moment matrix and finitely many unit-norm identities. Such a point need not come from a genuine probability measure. However, if a limit point satisfies \emph{every} finite consistency condition, then we will show that on each finite patch it does arise from a genuine probability measure. The marginals will be proven to be mutually consistent and hence can be glued together into a measure on the infinite lattice by the Kolmogorov extension theorem. Finally, it will follow easily that the measure is translation-invariant and that the energy density of this state is the true ground state energy density. \par

\textbf{Step 1: Extracting a limit point.} By Eq.~\eqref{eq:app_monotone}, the sequence $\varepsilon_{\cP_n,r_n}$ is monotone increasing and bounded above by $\varepsilon^*$. It therefore converges to some limit
\begin{equation}
    \varepsilon_\infty \le \varepsilon^*.
\end{equation}
For each $n$, the feasible set $\mathcal{F}_{\cP_n,r_n}\subset Y$ is nonempty because every translation-invariant state $\mu$ produces a feasible moment sequence $y^\mu \in \mathcal{F}_{\cP_n,r_n}$. Moreover, $\mathcal{F}_{\cP_n,r_n}$ is closed because each of its defining conditions defines a closed set, and it is hence the intersection of finitely many closed sets. Since the objective function $\langle \mathcal{E}\rangle_y$ is continuous on $Y$, it follows that the minimum defining $\varepsilon_{\cP_n,r_n}$ is attained. We may therefore choose, for each $n$, an optimizer
\begin{equation}
    y^{(n)} \in \mathcal{F}_{\cP_n,r_n}
\end{equation}
such that $\langle \mathcal{E}\rangle_{y^{(n)}} = \varepsilon_{\cP_n,r_n}$. Since $Y$ is compact, the sequence $\{y^{(n)}\}$ has a convergent subsequence. Passing to such a subsequence, which we denote by $\{y^{(n_k)}\}$, we obtain a point $y^* \in Y$ such that
\begin{equation}
    y^{(n_k)} \to y^*
\end{equation}
in $Y$. Because convergence in $Y$ is coordinatewise, this simply means that every moment coordinate converges:
\begin{equation}
    y^{(n_k)}_{[\alpha]} \to y^*_{[\alpha]}
\end{equation}
for each $[\alpha]\in \mathcal{C}$.
We now claim that
\begin{equation}\label{eq:ystarinintersection}
    y^* \in \bigcap_{\cP,r}\mathcal{F}_{\cP,r}.
\end{equation}
In other words, the limit point $y^*$ satisfies the bootstrap constraints for \emph{every} finite patch $\cP$ and hierarchy level $r$. \par

To see this, fix a particular finite patch $\cP$ and hierarchy level $r$. Since the sequence $(\cP_n,r_n)$ is assumed to grow to cover the whole lattice and to reach arbitrarily high hierarchy level, there exists some index beyond which 
\begin{equation}
    \cP \subseteq \cP_{n_k},
    \qquad
    r \le r_{n_k}.
\end{equation}
For those $k$, the optimizer $y^{(n_k)}$ must satisfy all of the constraints defining $\mathcal{F}_{\cP,r}$, and hence $y^{(n_k)} \in \mathcal{F}_{\cP,r}$ for all sufficiently large $k$. Now, passing to the limit $k\to\infty$ and using the closedness of $\mathcal{F}_{\cP,r}$, we have
\begin{equation}
    y^* \in \mathcal{F}_{\cP,r}.
\end{equation}
Since the choice of $\cP$ and $r$ was arbitrary, Eq.~\eqref{eq:ystarinintersection} follows. 

\textbf{Step 2: Reconstructing the state.}  
The conclusion of Step 1 was that the limit point $y^*$ satisfies the bootstrap constraints for every finite patch and every hierarchy level. We now show that this implies that $y^*$ is the moment sequence of a genuine state on $X$. \par

For each finite patch $\Lambda \subset \Gamma$, we show in Lemma~\ref{lem:finite_patch_measure} that one may construct a probability measure $ \mu_\Lambda$ on 
\begin{equation}
K_\Lambda \coloneqq (S^2)^\Lambda
\end{equation}
whose moments agree with the restriction of $y^*$ to $\Lambda$. Moreover, in Lemma~\ref{lem:consistency} we show that these measures are compatible under the operation of taking marginals: if $\Lambda' \subset \Lambda$, then the marginal of $\mu_\Lambda$ on $\Lambda'$ is exactly $\mu_{\Lambda'}$. This implies that the collection $\{\mu_\Lambda\}$ forms a \textit{projectively consistent} family of finite-dimensional distributions. \par

We may therefore apply the Kolmogorov extension theorem to our setting. This theorem yields a unique Borel probability measure $\mu$ on the full space $X$ whose marginal on each finite patch $\Lambda$ is $\mu_\Lambda$. We refer to, for example, \S36 of Ref.~\cite{billingsley1995probability} for a reference on this theorem. \par
It remains to show that $\mu$ is translation-invariant. Let $\vec t \in \mathbb{Z}^d$ and let $T_{\vec t} : X \to X$ denote translation by $\vec t$. Fix a finite patch $\Lambda \subset \Gamma$, and let $\vec t+\Lambda$ denote its translate. There is a natural relabeling map
\begin{equation}
    \theta_{\vec t,\Lambda} : K_{\vec t+\Lambda} \to K_\Lambda
\end{equation}
obtained by identifying the translated patch $\vec t+\Lambda$ with $\Lambda$. Now, we may compare the restriction of the translated state to the patch $\Lambda$ with the restriction of the original state to the translated patch $\vec t+\Lambda$. These are related by the relabeling map $\theta_{\vec t,\Lambda}$. Thus, for every polynomial $p$ on $K_\Lambda$, we have
\begin{align}
    \int_{K_\Lambda} p\, d\mu^{(\vec t)}_\Lambda
    &= \int_{K_{\vec t+\Lambda}} (p\circ \theta_{\vec t,\Lambda})\, d\mu_{\vec t+\Lambda}
    \nonumber\\
    &= \int_{K_\Lambda} p\, d\mu_\Lambda,
\end{align}
where $\mu_\Lambda$ denotes the marginal of $\mu$ on the patch $\Lambda$, and $\mu^{(\vec t)}_\Lambda$ denotes the marginal on the same patch $\Lambda$ of the translated state. The last step uses that $y^*$ is indexed by translation classes, so translating a monomial does not change its assigned moment. \par

Thus $\mu^{(\vec t)}_\Lambda$ and $\mu_\Lambda$ agree on all polynomial functions. By the Stone--Weierstrass theorem, the two measures agree on all continuous functions. Since this holds for every finite patch $\Lambda$, the translated state and the original state have the same finite-dimensional marginals everywhere, and therefore coincide. Hence $\mu$ is translation-invariant.

Finally, the moments of $\mu$ are exactly the entries of $y^*$. Indeed, let $[\alpha]\in \mathcal{C}$, and choose a representative $\alpha$ whose support is contained in some finite patch $\Lambda$. Then
\begin{equation}
    \int_X m_\alpha\, d\mu
    =
    \int_{K_\Lambda} m_\alpha \, d\mu_\Lambda
    =
    y^*_{[\alpha]}.
\end{equation}
Thus $y^*$ is the moment sequence of a genuine translation-invariant state on the infinite lattice.

\textbf{Step 3: Optimality of the limit point.}  
We now show that the energy density of $\mu$ is in fact the true ground state energy density $\varepsilon^*$.
\par 
Since the local energy density $\mathcal{E}$ is a finite linear combination of monomials, the corresponding energy functional $\langle \mathcal{E}\rangle_y$ is a continuous function of the moment coordinates. It follows that
\begin{equation}
    \langle \mathcal{E}\rangle_{y^{(n_k)}} \to \langle \mathcal{E}\rangle_{y^*}.
\end{equation}
On the other hand, by definition of the optimizers $y^{(n_k)}$,
\begin{equation}
    \langle \mathcal{E}\rangle_{y^{(n_k)}} = \varepsilon_{\cP_{n_k},r_{n_k}} \to \varepsilon_\infty.
\end{equation}
Therefore $ \langle \mathcal{E}\rangle_{y^*} = \varepsilon_\infty$. In Step 2 we constructed a state $\mu$ such that
\begin{equation}
    \langle \mathcal{E}\rangle_{y^*}
    =
    \int_X \mathcal{E}\, d\mu.
\end{equation}
By definition of $\varepsilon^*$ as the infimum of $\int_X \mathcal{E}\, d\mu$ over all $\mu\in M_{\mathrm{TI}}(X)$, this implies
\begin{equation}
    \varepsilon_\infty
    =
    \int_X \mathcal{E}\, d\mu
    \ge \varepsilon^*.
\end{equation}
From Step 1 we recall $\varepsilon_\infty \le \varepsilon^*$, and hence we may conclude $ \varepsilon_\infty =  \lim_{n\to\infty} \varepsilon_{\cP_n,r_n} = \varepsilon^*$. This completes the proof.
\end{proof}

\section{\blue{The Heisenberg--Kitaev model}}
\label{app:hk}

\begin{table}[!htb]

\caption{\blue{\textit{Numerical settings and run times.} Patch size $|\cP|$, hierarchy level $r$, and solve time on one CPU core for the bilinear--biquadratic and Heisenberg--Kitaev models. All SDPs were solved using \texttt{MOSEK}~\cite{mosek} through JuMP~\cite{Lubin2023}, with a relative duality-gap tolerance of $10^{-7}$.}}
\label{tab:sdp}
\begin{ruledtabular}
\begin{tabular}{lccr}
\blue{Model} & \blue{$|\cP|$} & \blue{$r$} & \blue{wall time} \\
\hline
\blue{Bilinear--biquadratic~\eqref{eq:blbq_hamiltonian}} & \blue{$4$} & \blue{$2$} & \blue{$1$--$5$ s} \\
\blue{Heisenberg--Kitaev~\eqref{eq:hk_hamiltonian}}      & \blue{$4$} & \blue{$2$} & \blue{$7$ s} \\
\blue{Heisenberg--Kitaev}                                & \blue{$5$} & \blue{$2$} & \blue{$25$ s} \\
\blue{Heisenberg--Kitaev}                                & \blue{$6$} & \blue{$2$} & \blue{$2.4$ min} \\
\blue{Heisenberg--Kitaev}                                & \blue{$7$} & \blue{$2$} & \blue{$30$ min} \\
\end{tabular}
\end{ruledtabular}
\end{table}

\begin{figure}[t]
    \centering
    \includegraphics[width=\columnwidth]{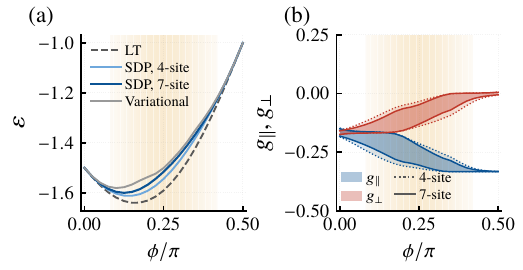}
\caption{\blue{\textit{Heisenberg--Kitaev model on the triangular lattice.} We set $J=\cos\phi$ and $K=\sin\phi$, with the shaded region indicating the frustrated regime. (a) Ground state energy density: Luttinger--Tisza bound (gray dashed), SDP lower bounds at $r=2$ for four-site (light blue) and seven-site (dark blue) patches, and variational upper bound (gray). (b) Bounds on the bond-resolved correlators $g_\parallel$ (blue) and $g_\perp$ (red) of Eq.~\eqref{eq:hk_correlators}, for both patch sizes.}}

    \label{fig:hk}
\end{figure}

\blue{To explore how the bootstrap handles spin anisotropy, we consider the classical Heisenberg--Kitaev model on the triangular lattice~\cite{becker2015spinorbit,rousochatzakis2016kitaev}. This model is motivated by bond-dependent interactions in spin-orbit-coupled Mott insulators~\cite{jackeli2009mott}, and has Hamiltonian}
\begin{equation}
    H = \frac{1}{2}\sum_{\langle ij\rangle}
    \bigl[\cos\phi\, \vec{S}_i \cdot \vec{S}_j
    + \sin\phi\, S_i^{\gamma_{ij}} S_j^{\gamma_{ij}} \bigr],
\label{eq:hk_hamiltonian}
\end{equation}

\blue{Here $\gamma_{ij}\in\{x,y,z\}$ labels the spin component associated with each bond direction. The model interpolates between the $120^\circ$ antiferromagnet at $\phi=0$ and a collinear state with $\varepsilon=-1$ at $\phi=\pi/2$~\cite{becker2015spinorbit}. Since the bond-dependent interactions reduce the continuous $O(3)$ symmetry to a discrete subgroup, we must retain correlations between individual spin components rather than only their scalar products. The bootstrap construction itself is otherwise unchanged.}

\par

\blue{Because the lattice is Bravais and the Hamiltonian is quadratic, we can compare directly with Luttinger--Tisza, unlike in the quartic example of Sec.~\ref{sec:bbq}. Figure~\ref{fig:hk}(a) shows the LT bound alongside the $r=2$ SDP lower bounds and the variational upper bound. At $r=1$, the SDP agrees with LT to within $3.3\times10^{-6}$ throughout the parameter range, as expected from Sec.~\ref{sec:LT}. The improvement at $r=2$ therefore comes from the additional higher-order constraints.}
\par

\blue{We also find that enlarging the patch at fixed $r=2$ systematically improves the energy bounds. The four-site patch contains two triangles sharing a bond, while the seven-site patch contains a central site and its six nearest neighbors, allowing us to include constraints involving all six surrounding triangles. As shown in Fig.~\ref{fig:hk}(a), this substantially narrows the gap between the SDP lower bound and the variational upper bound across the phase diagram.}
\par

\blue{The anisotropy also allows us to test bounds on correlations that distinguish the spin components along each bond. We consider the nearest-neighbor correlations parallel and perpendicular to the corresponding axis,}
\begin{equation}
    g_\parallel \coloneqq \frac{1}{6}\sum_{d} \bigl\langle S_0^{\gamma_d} S_d^{\gamma_d}\bigr\rangle,
    \qquad
    g_\perp \coloneqq \frac{1}{12}\sum_{d}\sum_{a \neq \gamma_d} \bigl\langle S_0^{a} S_d^{a}\bigr\rangle,
\label{eq:hk_correlators}
\end{equation}
Here $d$ runs over the six nearest-neighbor displacements, and $\gamma_d$ denotes the spin axis associated with each bond. Figure~\ref{fig:hk}(b) shows the resulting bounds from Eq.~\eqref{eq:observable_sdp}. At the Heisenberg point $\phi=0$, both correlators approach $-1/6$, while at the Kitaev point $\phi=\pi/2$ they approach $g_\parallel=-1/3$ and $g_\perp=0$. Away from these limits, neither correlator is fixed by symmetry. Increasing the patch size tightens the bounds throughout the phase diagram, reducing their mean widths from $0.068$ to $0.053$ for $g_\parallel$, and from $0.066$ to $0.051$ for $g_\perp$. These results demonstrate that the bootstrap remains effective for anisotropic interactions, with the main cost being a larger moment basis.

\end{document}